\documentclass[12pt,onecolumn,journal]{IEEEtran}
\usepackage{citesort}
\usepackage[dvips]{graphicx}
\usepackage[cmex10]{amsmath}
\usepackage{algorithm}
\usepackage{algpseudocode}
\usepackage{array}
\usepackage[tight,footnotesize]{subfigure}
\usepackage{setspace}
\usepackage{enumerate} 
\usepackage{url}

\usepackage{amsfonts, amssymb, amsthm, mathrsfs}
\usepackage{setspace}
\usepackage{Tabbing}
\usepackage{fancyhdr}
\usepackage{lastpage}
\usepackage{extramarks}
\usepackage{chngpage}
\usepackage{soul,color}
\usepackage{indentfirst}
\usepackage{float,wrapfig}
\usepackage{times}
\usepackage{multirow}
\usepackage{bm}

\newtheorem{theo}{Theorem}
\newtheorem{defi}{Definition}
\newtheorem{assu}{Assumption}
\newtheorem{prop}{Proposition}

\newcolumntype{I}{!{\vrule width 1.5pt}}
\newlength\savedwidth

\begin{document}

\title{User Association for Load Balancing in Heterogeneous Cellular Networks}

\author{Qiaoyang Ye,
        Beiyu Rong,
        Yudong Chen,
        Mazin Al-Shalash,
        Constantine Caramanis
        and Jeffrey G. Andrews
\thanks{Q. Ye, B. Rong, Y. Chen, J. G. Andrews and C. Caramanis are with WNCG, The University of Texas at Austin, USA, M. Shalash is with Huawei Technologies. Email: 
qye@utexas.edu, byrong@austin.utexas.edu, ydchen@utexas.edu,  mshalash@huawei.com, \{cmcaram, jandrews\}@ece.utexas.edu. This research has been supported by the National Science Foundation CIF-1016649 and Huawei. Manuscript last revised: \today}}

\maketitle
\begin{abstract}
For small cell technology to significantly increase the capacity of tower-based cellular networks, mobile users will need to be actively pushed onto the more lightly loaded tiers (corresponding to, e.g., pico and femtocells), even if they offer a lower instantaneous SINR than the macrocell base station (BS).  Optimizing a function of the long-term rates for each user requires (in general) a massive utility maximization problem over all the SINRs and BS loads.  On the other hand, an actual implementation will likely resort to a simple biasing approach where a BS in tier $j$ is treated as having its SINR multiplied by a factor $A_j \geq 1$, which makes it appear more attractive than the heavily-loaded macrocell.  This paper bridges the gap between these approaches through several physical relaxations of the network-wide association problem, whose solution is NP hard.  We provide a low-complexity distributed algorithm that converges to a near-optimal solution with a theoretical performance guarantee, and we observe that simple per-tier biasing loses surprisingly little, if the bias values $A_j$ are chosen carefully.   Numerical results show a large (3.5x) throughput gain for cell-edge users and a 2x rate gain for median users relative to a maximizing received power association.
\end{abstract}


\IEEEpeerreviewmaketitle

\section{Introduction}
To meet surging traffic demands, cellular networks are trending strongly towards increasing heterogeneity, especially through proliferation of small BSs, e.g., picocells and femtocells, which differ primarily in terms of maximum transmit power, physical size, ease-of-deployment and cost \cite{DhiAnd12,Lag97,SalRus87,SydTao09,WuMur04,AndCla12}. Heterogeneous networks (HetNets) enable a more flexible, targeted and economical deployment of new infrastructure versus tower-mounted macro-only systems, which are very expensive to deploy and maintain \cite{DamMon11}.  Even with a targeted deployment where these small BSs are placed in high-traffic zones, most users will still receive the strongest downlink signal from the tower-mounted macrocell BS.  In order to make the most of the new low-power infrastructure, mobile users should be actively ``pushed'' onto the small BSs, which will often be lightly loaded and so can provide a higher rate over time by offering the mobile many more resource blocks than the macrocell.  Similarly, a more balanced user association reduces the load on the macrocell, allowing it to better serve its remaining users.  This paper investigates optimal and near-optimal solutions of this cell association problem, particularly those with simple requirements for coordination and side information.

\subsection{Related Work}
Most prior work on load balancing schemes applies mostly to macrocell-only networks.  HetNets are much more sensitive to the cell association policy because of the massive disparities in cell sizes.  These unequal cell sizes result in very unequal loads in a max-SINR cell association, assuming a relatively uniform mobile user distribution.  That is, if users simply associate with the strongest BS, the difference in load in macrocell networks is constrained since the cells all have roughly the same coverage area.  But in HetNets, the opposite is true, making the problem considerably more complex, and the potential gains from load-aware associations larger. 

The existing work on cell association can be broadly classified into two groups:
\begin{enumerate}
\item Strategies based on \textit{channel borrowing} from lightly-loaded cells, such as hybrid channel assignment (HCA) \cite{KahGeo78}, channel borrowing without locking (CBWL) \cite{JiaRap94}, load balancing with selective borrowing (LBSB) \cite{DasSen97,DasSen98}, etc;
\item Strategies based on \textit{traffic transfer} to lightly-loaded cells, such as directed retry \cite{Ekl86}, mobile-assisted call admission algorithms (MACA) \cite{WuMuk00}, hierarchical macrocell overlay systems \cite{CavAgr05,YanTon04}, cell breathing techniques \cite{DasVis03,BejHan09}, and biasing methods in HetNets \cite{DamMon11}.
\end{enumerate}
The approach in this paper is based on traffic transfer.  There have been many efforts in the literature toward traffic transfer strategies in macro-only cellular networks. The so-called ``cell breathing'' technique \cite{DasVis03,BejHan09} dynamically changes (contracts or expands) the coverage area depending on the load situation (over-loaded or under-loaded) of the cells by adjusting the transmit power. Sang \textit{et al.} \cite{SanWan08} proposed an integrated framework consisting of MAC-layer cell breathing and load-aware handover/cell-site selection. Cell breathing aims to balance the load among neighboring macrocells, while in HetNets we additionally need to balance the load among different tiers.

A popular approach in conventional networks, related to the direction we propose, is to achieve load balancing by changing the objective function to be concave. Indeed, there is considerable work investigating different utility functions, such as network-wide proportional fairness \cite{BuLi06}, network-wide max-min fairness \cite{BejHan07}, maximization of network-wide aggregate utility by partial frequency reuse and load balancing \cite{SonCho09}, and $\alpha$-optimal user association \cite{KimVec11}. We adopt the logarithmic function as the utility function, which is similar to proportional fairness, and achieves a desirable tradeoff between opportunism and fair allocation across users, by saturating the reward for providing more rate to users which already have a high rate.  

In HetNets, there are a few recent investigations of the cell association problem. A joint optimization of channel selection, user association and power control in HetNets is considered in \cite{CheBac11}, aiming to minimize the potential delay, which is related to the sum of the inverse of the per-user SINRs, where the SINR takes into account the load when computing the interference. Corroy \textit{et al.} \cite{CorFal12} propose a dynamic cell association to maximize sum rate as well as a heuristic cell \textit{range expansion} algorithm for load balancing. Cell range expansion is an effective method to balance the load among high and low power BSs, which is enabled through cell biasing~\cite{DamMon11,JoSan11}. It is achieved by performing user association based on the biased measured signal, which leads to better load balancing, but the improvement of load balancing may not overwhelm the degradation in SINR that certain users suffer. Therefore, how to design the biasing factor is an important open problem. 

\subsection{Contributions and Organization}
We present a load-aware cell association method and distributed algorithm for downlink HetNets, that results in the following main contributions.

First, in Section~\ref{sec:formulation}, we undertake an optimization theoretic approach to the load-balancing problem, where we consider cell association and resource allocation jointly. We decouple the joint general utility maximization problem by assuming (optimistically) that users can be associated with more than one BS. This approach provides an upper bound on achievable network utility which can serve as a benchmark. However, in real system, it is much more difficult to implement multi-BS association than single-BS association. Therefore, we formulate a logarithmic utility maximization problem for single-BS association, and show that equal resource allocation is actually optimal, over a sufficiently large time window. This observation allows the coupled problem in single-BS association to reduce to the cell association problem with equal resource allocation, which along with the fractional association assumption converts the previously intractable combinatorial problem into a convex optimization problem. 

In Section \ref{sec:algo}, we exploit the convexity of the problem to develop a distributed algorithm via dual decomposition that converges towards the optimal solution with a guarantee on the maximum gap from optimality. This provides a feasible, efficient and low-overhead algorithm for implementation in HetNets. 

In Section \ref{sec:insights}, we leverage our provably optimal solutions to ask a basic question: how much of the performance gain can a simple policy based on \emph{a priori} bias factors achieve? Our results show that this simple approach gets surprisingly close to the gains of the load-aware utility maximization.  The gains from this approach are shown to be very large for most users in the system, with rate gains ranging from 2-3.5x for the bottom half of users.  To put this in context, this is a gain on par with what would otherwise be achieved by a doubling or tripling the amount of spectrum for a given service provider.  Cell interior users experience little to no rate gain (or a small loss), but this has little relevance in practice since such users are already well-served.

\section{System Model} \label{sec:model}
In downlink (DL) cellular networks, the default association scheme is max-SINR, which indeed maximizes the probability of coverage, i.e., $\mathbb{P}(\text{SINR}>\beta)$, where $\beta$ is a target SINR (or equivalently minimizes the probability of outage, i.e., $\mathbb{P}(\text{SINR}\leq\beta)$). In conventional networks, the default association scheme in uplink (UL) is typically the same as the association in DL, since the coverage areas are almost the same among different macrocells. However, this is not the case in HetNets, since the BSs of different tiers have such widely divergent transmit powers. Rather, the DL coverage area of macro BSs is much larger than that of smaller BSs. If we adopt the same association in UL, the cell-edge macro-users will cause great interference to nearby users, especially for users which are associated to nearby small cells. Furthermore, the key metric for performance is their service \emph{rate}, not \emph{SINR}.  The instantaneous rate is of course directly related to SINR (e.g., $\log_2(1 + \textrm{SINR})$),  but the overall served rate is then multiplied by the fraction of resources that user gets.  Hence, heavily-loaded cells provide lower rate over time, even if they provide a higher SINR. Load balancing problem is very important in both DL and UL HetNets.

%

In this paper, we focus on DL cell association. UL could likely be considered through a similar approach, but is complicated by the use of UL power control, which changes the interference depending on the association.  Here, we assume that all BSs have full buffers and slowly changing (or constant) transmit power, which means that BSs are always sending at a fixed transmit power over the association time scale, and thus the interference power is a constant which is independent of the specific association.

A downlink HetNet consisting of $K$-tiers of BSs is illustrated in Fig. \ref{fig:hetnet} with $K=3$. Each tier models a particular type of BS: for example, tier 1 consists of traditional macrocells, and tier 2 and tier~3 could be interpreted as being comprised of picocells and femtocells, respectively. Picocells transmit at a lower power with a higher deployed density than the macrocells, while the femtocells, or the home BSs, may eventually be deployed very densely but have a very small transmit power.

\begin{figure}[!hbp]
\centering
\includegraphics[scale=0.6]{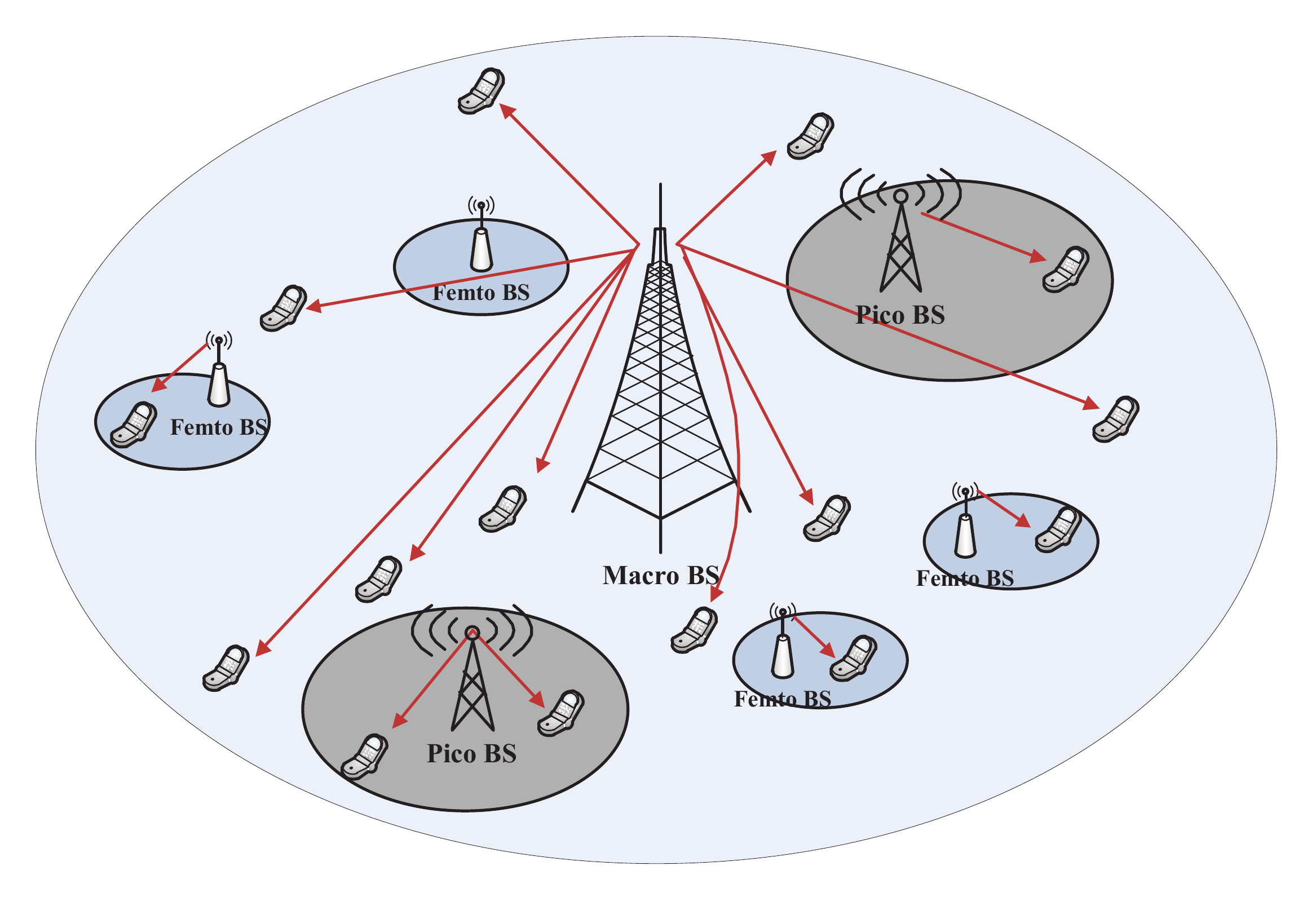}
\caption{Illustration of a three-tier heterogeneous cellular network. Only a single macro-cell is shown for simplicity.}
\label{fig:hetnet}
\end{figure}

We denote by $\mathcal{B}$ the set of all BSs, and $\mathcal{U}$ the set of all users. During the connection period, we denote by $c_{ij}$ the achievable rate, where generally, $c_{ij}$ is a logarithmic function of SINR. 
\begin{eqnarray}
c_{ij}=f(\textrm{SINR}_{ij})=f(\frac{P_j g_{ij}}{\sum\limits_{k\in \mathcal{B}, k\neq j} P_k g_{ik}+\sigma^2}),
\end{eqnarray}
where $P_j$ is the transmit power of BS $j$, $g_{ij}$ denotes the channel gain between user $i$ and BS~$j$, which in general includes path loss, shadowing and antenna gain, and $\sigma^2$ denotes the noise power level. The association is carried out in a large time scale compared to the change of channel. The SINR for association is averaged over the association time and thus it is a constant regardless of the dynamics of channels. As for resource allocation, we assume that resource allocation is carried out well during the channel coherence time, and thus channel can be regarded as static during each resource allocation period. This model is applicable for low mobility environment. We leave the stochastic channel analysis as our future work. Note that though this paper is focus on single-carrier system, our model can be extended to multi-carrier system in a straightforward manner (i.e., let $c_{ij}$ be the average rate over different bands).

Since each BS generally serves more than one user, users of the same BS need to share resources such as time and frequency slots. The long-term service rate experienced by a user thus depends on the load of the BS and will therefore be only a fraction of the value $c_{ij}$ (unless BS $j$ exclusively serves user $i$). We assume that users will keep transmitting during an association time scale (i.e., users have full buffers), so the load on a BS is directly proportional to the total number of users associated with it.

Moreover, the overall service rate also depends on the resource allocation method of the BSs. In principle, any allocation method or service discipline with which the resource allocation is related to both the load of BSs and the rate of each user can be used. Therefore, the achievable overall rate of user $i$ associated with BS $j$ depends on $c_{ij}, c_{qj}$, and how BS $j$ distributes its resources among its associated users, for all $q\in\mathcal{U}\setminus{i}$.  We focus on finding an optimal resource allocation and optimal cell associations  which maximize the utility. 
During the connection between the BS $j$ and user $i$, denoting the fraction of resource BS serves user $i$ by $y_{ij}$, we can define the overall long term rate as follows.
\begin{defi}
If user $i$ is associated with BS $j$, the overall long term rate is
\vspace{-0.4cm}
\begin{equation}
R_{ij}=y_{ij}c_{ij},\vspace{-0.5cm}
\end{equation}
where $\sum_i y_{ij}=1, \ \forall j$. We denote the total overall rate of user $i$ as $R_i$, where $R_i=\sum_j R_{ij}$.
\end{defi}
In the following, we investigate a utility maximization problem for the overall rate $R_{i}$ to find the optimal association and resource allocation.

\section{Problem Formulation}\label{sec:formulation}
Taking a utility function perspective, we assume user $i$ obtains utility $U_{i}(R_{i})$ when receiving rate is $R_{i}$, where the function $U_{i}(\cdot)$ is a continuously differentiable, monotonically increasing, and strictly concave utility function~\cite{StaWic09}.
\subsection{General Utility Maximization: Unique Association}
 We formulate an optimization problem which involves finding the indicators $\left\lbrace x_{ij}\right\rbrace$ corresponding to the association (i.e., $x_{ij}=1$ when user $i$ is associated with BS $j$, otherwise $x_{ij}=0$) and $\{y_{ij}\}$ corresponding to the resource allocation that maximizes the aggregate utility function:
\vspace{-0.1cm}
\begin{equation}\label{eq:general}
\begin{aligned}
 \max\limits_{x,y} & \quad \sum\limits_{i\in \mathcal{U}} U_{i}(R_i)=\sum\limits_{i\in \mathcal{U}} U_{i}(\sum\limits_{j\in\mathcal{B}} y_{ij}c_{ij})\\
 \textrm{s.t.} & \quad \sum\limits_{j\in\mathcal{B}} x_{ij}=1, \quad \forall i\in \mathcal{U} \\
 &\quad \sum\limits_{i\in\mathcal{U}} y_{ij}\leq1, \quad\forall j\in\mathcal{B}\\
 & \quad  0\leq y_{ij}\leq x_{ij},\  x_{ij} \in \{0,1\} \quad \forall i\in \mathcal{U}, \forall j\in \mathcal{B}.
\end{aligned}
\end{equation}

\subsection{General Utility Maximization: Allowing Joint Association}
\label{sec:III-B}
The indicator variable $x_{ij}$ enforces unique association, which is combinatorial. Moreover, the cell association has to be considered jointly with resource allocation, because resource allocation depends on the association and user association depends on the achievable resource for each user. Therefore, the resulting problem is difficult to solve. While allowing a user to be served by multiple BSs may require more overhead to implement, and hence perhaps may not be viable in practice, it provides an upper bound on the network performance. In this section, we make the following assumption:
\begin{assu}\label{as:more}
We assume that users can be associated with more than one BS at the same time. 
\end{assu}

Under this assumption,  the constraint $\sum_j x_{ij}=1$ can be eliminated, and hence there is no need for  $x_{ij}$ as additional indicators for cell association. The resource allocation variable $y_{ij}\in [0,1]$ indicates the association, i.e., user $i$ is associated with BS $j$ when $y_{ij}>0$, otherwise they are not connected.

Therefore, we focus only on investigation of how the resource should be allocated to different users with different rate $c_{ij}$ so as to maximize the utility, instead of considering in conjunction with cell association. 

We formulate the \textit{joint association} problem as follows:
\begin{equation}\label{eq:jointop}
\begin{aligned}
     \max_y &\quad \sum_{i\in\mathcal{U}} U_i(\sum_j y_{ij}c_{ij}) \\
     \textrm{s.t.} & \quad    \sum_{i\in\mathcal{U}}y_{ij}\leq 1,\quad \forall j\in\mathcal{B}\\
     & \quad 0\leq y_{ij} \leq 1, \forall i\in \mathcal{U} ,\forall j\in \mathcal{B}.
\end{aligned}            
\end{equation}
Note that this joint association scheme focuses on how to allocate resource for each BS, rather than how to associate users. In the following sections, we show that with some specific utility functions (e.g., logarithmic utility),~$y_{ij}$ can be directly found without Assumption \ref{as:more} and thus there is no need to decouple $x_{ij}$ and $y_{ij}$ as in this optimization. However, problem (\ref{eq:jointop}) provides an ultimate limit on achievable network performance for general utility maximizations. Interestingly, our simulation results show that the bound is quite tight in logarithmic utility maximization.

\subsection{Logarithmic Utility Formulation}
Using linear utility functions for throughput maximization results in a trivial solution, where each BS serves only its strongest user. While throughput-optimal, this is not a satisfactory solution for many reasons. Instead, we seek a utility for that naturally achieves load balancing, and some level of fairness among the users. To accomplish this, we use a logarithmic utility function. The resulting objective function with logarithmic utility is
\begin{equation}\label{eq:sumlog formulation}
U_{i}(R_{i})= \log\left(\sum_j y_{ij}c_{ij}\right). 
\end{equation}
This logarithm is concave, and hence has diminishing returns. This property encourage load balancing. This is consistent with the resource allocation philosophy in real systems, where allocating more resources for a well-served user is considered low priority, whereas providing more resources to users with low rates (e.g., in the linear region of the logarithmic function) is considered desirable. Logarithmic function in particular is a very common choice of utility function. Therefore, in the remainder of this paper, we use a logarithmic utility function.



\subsection{Analysis of Optimized Resource Allocation}
For general utility functions, we proposed one possible tractable model for the joint cell association and resource allocation problem in Set. \ref{sec:III-B}, which allows users to be served by multiple BSs. In practice, this is much more difficult to implement than single-BS association. Therefore, we consider it as a benchmark in this paper, providing an upper bound on network performance. In the remainder of this paper, we focus on the log utility function in single-BS association. 

In single-BS association, the objective function of ($\ref{eq:general}$) becomes 
\begin{equation}
\sum\limits_{j\in \mathcal{B}} \sum_{i\in\{l| x_{lj}=1\}} U_{i}( y_{ij}c_{ij}).
\end{equation}

Then, we conduct the resource allocation analysis on a typical BS $j$ and the users associated with that BS. The utility maximization problem for the users associated with BS $j$ is
\begin{equation}\label{eq:resourceallo}
\begin{aligned}
     \max_y &\quad \sum_{i\in\{l| x_{lj}=1\}}  \log ( y_{ij}c_{ij}) \\
     \textrm{s.t.} &\quad \sum\limits_{i\in\mathcal{U}} y_{ij} \leq 1, \\
&\quad 0\leq y_{ij}\leq 1 \quad \forall i\in \mathcal{U}.
\end{aligned}            
\end{equation}
\begin{defi}
We define the effective load of BS $K_j$ as the number of users associated with it, i.e., $K_j=\sum\limits_{k\in \mathcal{U}} x_{kj}$,
where $x_{ij}$ is the association indicator.
\end{defi}
The optimization (\ref{eq:resourceallo}) suggests the following proposition.
\begin{prop}\label{prop:equalresource}
The optimal resource allocation is equal allocation, i.e., $y_{ij}=\left.1\middle/{K_j}\right.$.
\end{prop}
\begin{proof}
See Appendix ~\ref{proof:equalresource}.
\end{proof}
In this paper, we assume the channel is static during each resource allocation period. Note that in stochastic setting which takes into account time-varying channels and user mobility, this proposition becomes exactly proportional fair scheduling (i.e., to maximize the log utility in terms of long-term average throughput).

Therefore, using the logarithmic utility function, the resource allocation is quite simple, and is independent of the distribution of SINR, which makes the joint cell association and resource allocation problem tractable. In particular, the optimal allocation is uniform across users served by that BS.

Given this equal resource allocation, the long-term rate for user $i$ from BS $j$ is
\vspace{-0.3cm}
\begin{equation}\label{eq:R}
R_{ij}= \frac{c_{ij}}{K_j},
\end{equation}
\noindent so we can rewrite the optimization (\ref{eq:general}) as
\begin{equation}\label{eq:ingerop}
\begin{aligned}
 \max\limits_x & \quad \sum\limits_{i\in\mathcal{U}} \sum\limits_{j\in\mathcal{B}} x_{ij} \log\left(\frac{c_{ij}} {\sum_k x_{kj}} \right)\\
 \textrm{s.t.} & \quad \sum\limits_{j\in\mathcal{B}} x_{ij}=1, \quad \forall i\in \mathcal{U}, \\
 & \quad x_{ij} \in \{0,1\}, \quad \forall i\in \mathcal{U}, \textrm{ and } \forall j\in \mathcal{B}.
\end{aligned}
\end{equation}

When the network is small, the optimal user association can be found through a brute force search. As an illustrative example, Fig. \ref{fig:largeex} compares the resulting association patterns of max-SINR vs. the proposed load-aware association scheme in (\ref{eq:sumlog formulation}). In Fig. \ref{maxSINRl}, max-SINR associates many users with macro BS $1$, which overloads it; on the other hand, many of the small BSs serve very few users with some even being idle. Fig. \ref{roundingl} shows the load-aware association, which moves traffic off congested macrocells and onto more lightly loaded small cells. Note that admission control carries out a similar task, where new arrival users will be blocked or forced to other lightly loaded BSs when the potential BS is heavily loaded. However, admission control is performed before a connection is established (i.e., only for new users rather than existing users), and thus cannot achieve an optimal association in terms of load balancing. 
\begin{figure}[!hbp]
\centering
\setcounter{subfigure}{0}
\subfigure[max-SINR association]{\label{maxSINRl}\includegraphics[width=8cm, height=6cm]{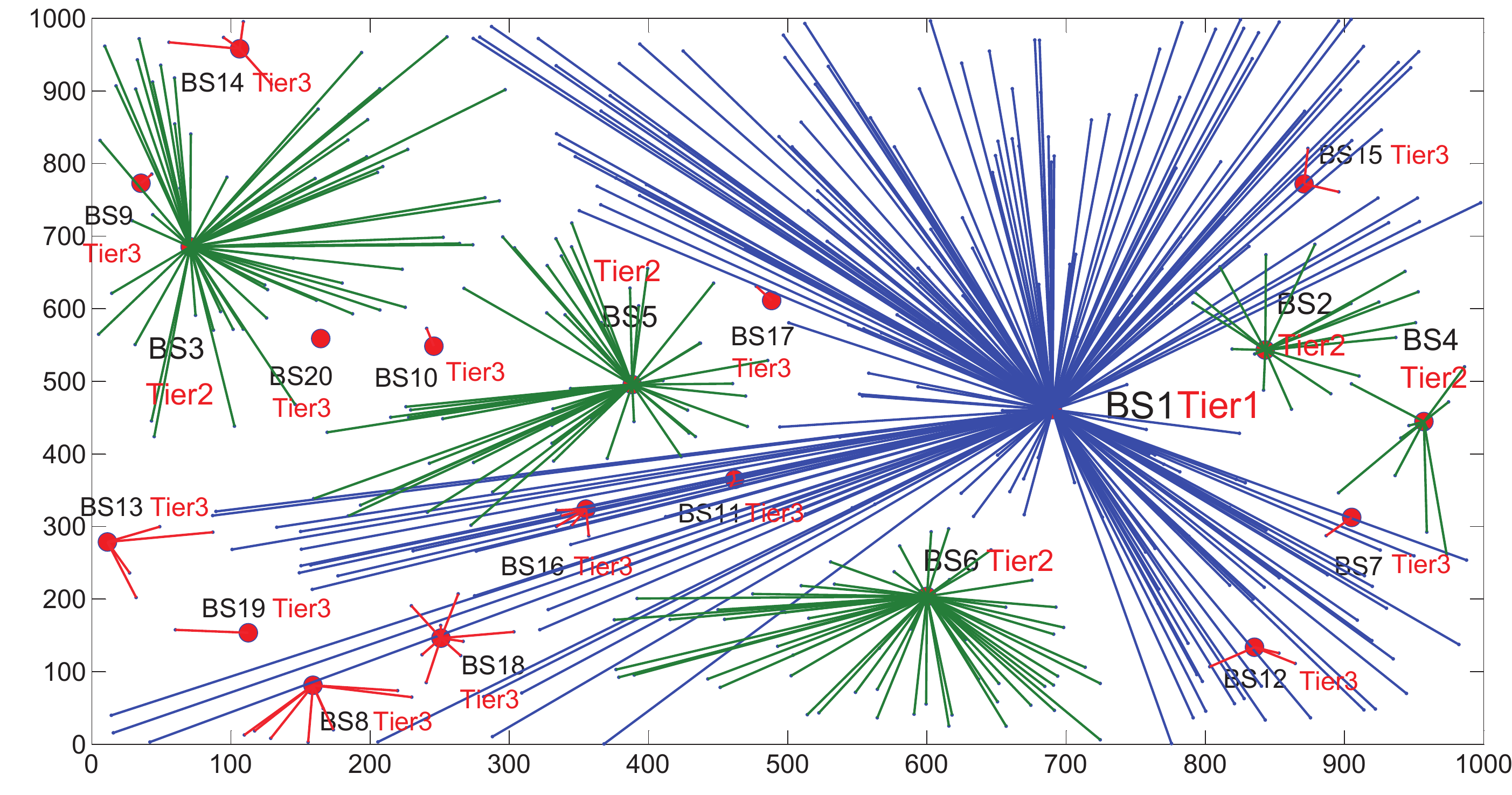}}
 \hspace{0.1in}
\subfigure[Fractional-rounding methods with max-sum-log objective]{\label{roundingl}\includegraphics[width=8cm, height=6cm]{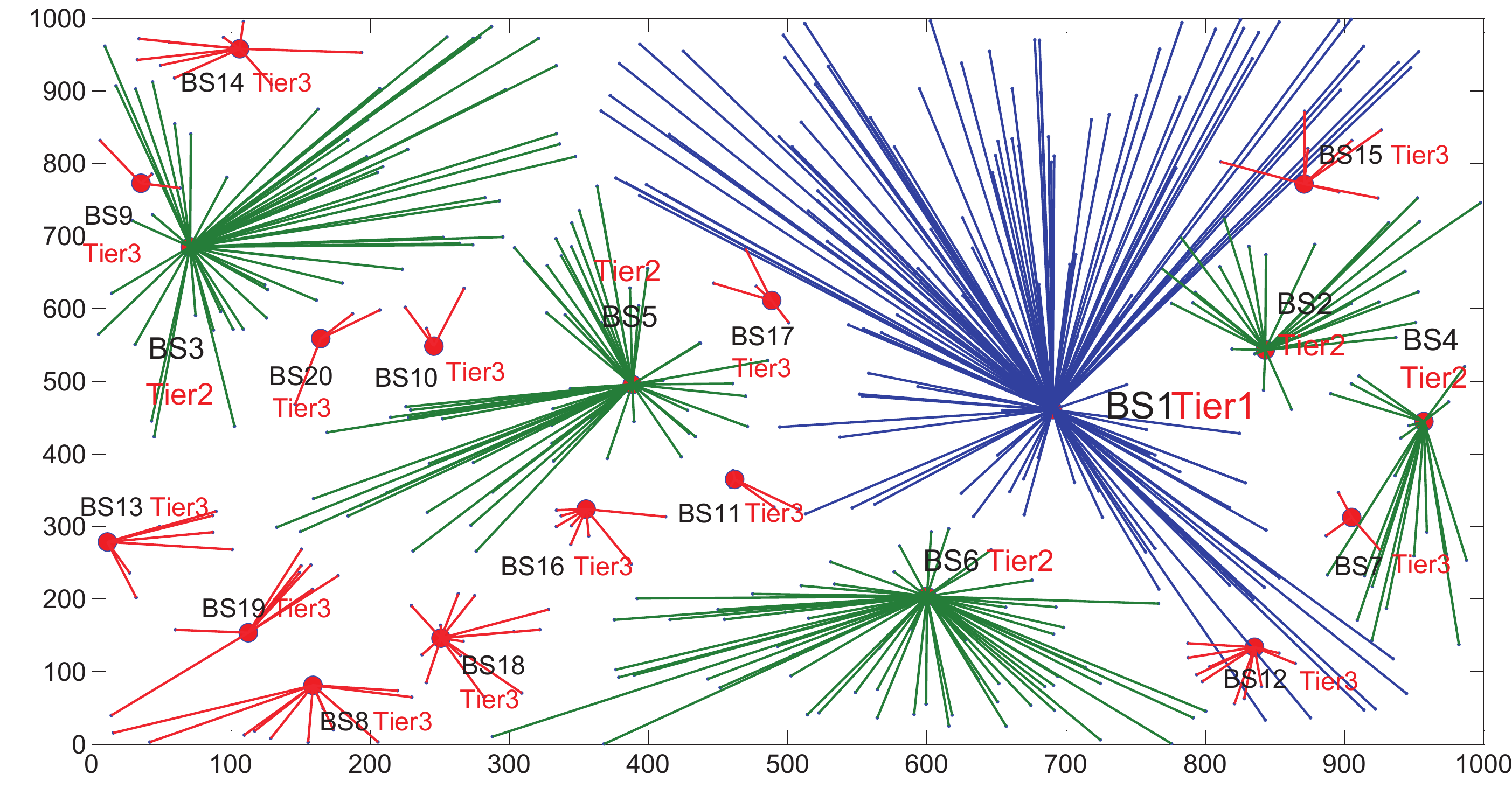}}
\caption{Different associations in HetNets. The load-aware scheme leads to more efficient resource utilization, by handover of load to underutilized small BSs}
\label{fig:largeex}
\end{figure}
\subsection{Relaxation to Fractional User Association}
The above problem is combinatorial due to  the binary variable $x_{ij}$. The complexity of the brute force algorithm is $\Theta((N_B)^{N_U})$, where $N_B$ and $N_U$ denote the number of BSs and number of users, respectively. The computation is essentially impossible for even a modest-sized cellular network. To overcome this, we again invoke Assumption \ref{as:more} to allow users to be associated to more than one BSs, i.e., ``fractional user association'' (FUA). This physical relaxation reduces the complexity which is no longer combinatorial, and upper bounds the special case where each user is associated with just one BS. It is more difficult to implement multiple-BS association than single-BS association in a practical system, and thus we adopt a rounding method to revert to single BS association~(\ref{eq:ingerop}).

Note that the upper bound provided by FUA is different from that provided by joint association with a general utility function: FUA upper bounds the performance of the cell association problem with equal resource allocation, while the joint association problem on Section \ref{sec:III-B} provides an upper bound without any restriction on resource allocation.  Nevertheless, the numerical results in Section \ref{sec:simulation} show that there is almost no loss after rounding, and the upper bound provided by FUA is quite tight.

With FUA, the indicators $x_{ij}$ can take on any real value in~$[0,1]$. The following relaxation of (\ref{eq:general}) is convex:
\begin{equation}\label{eq:convexop}
\begin{aligned}
 \max_x & \quad \sum\limits_{i\in\mathcal{U}} \sum\limits_{j\in\mathcal{B}} x_{ij}\log\left(\frac{c_{ij}}{\sum_{k\in\mathcal{U}} x_{kj}}\right)\\
 \textrm{s.t.} & \quad \sum\limits_{j\in\mathcal{B}} x_{ij}=1, \quad \forall i\in \mathcal{U}, \\
 & \quad 0 \leq x_{ij} \leq 1, \quad \forall i\in \mathcal{U}, \textrm{ and } \forall j\in \mathcal{B}.
\end{aligned}
\end{equation}


%

To directly solve the convex optimization (\ref{eq:convexop}), global network information is necessary, which requires a centralized controller for user association and coordination. In the following section, we propose a distributed algorithm without coordination.

%
%

\section{Primal-Dual Distributed Algorithm}\label{sec:algo}
The centralized functionality for solving the convex optimization problem is usually implemented by a server in the core network for macrocells (e.g., Radio Network Controller (RNC) which carries out resource management in UMTS), by only allowing slow adaptation at relatively long timescales and requirement of coordination among different tiers. Additional issues with centralized mechanisms include excessive computational complexity and low reliability, as any crash on the centralized controller operation will disrupt load balancing. In HetNets, it is usually difficult to coordinate macrocellls and femtocells which are deployed by operators and users respectively. Therefore, a low complexity distributed algorithm without coordination is desirable.

In this section, we propose a distributed algorithm via Lagrangian dual decomposition \cite{LowLap99}. The dual problem of (\ref{eq:convexop}) is decoupled into two sub-problems, which can be solved separately on users' side and BSs' side respectively.

\subsection{Dual Decomposition}
The primal formulation in (\ref{eq:convexop}) can be expressed in an equivalent form by introducing a new set of variables,  the load metric $K_j=\sum\limits_i x_{ij}$.
\vspace{-0.45cm}
\begin{equation}\label{primal}
\begin{aligned}
 \max_x & \quad \sum\limits_i \sum\limits_j x_{ij}\log\left({c_{ij}}\right)-\sum\limits_j K_j \log\left(K_j\right)\\
 \textrm{s.t.} & \quad \sum\limits_j x_{ij}=1, \quad \forall i\in U \\
 & \quad \sum\limits_i x_{ij}=K_j, \quad \forall j \in B\\
 & \quad K_j\leq N_U\\
 & \quad x_{ij}, K_j \geq 0, \quad \forall i\in U, \textrm{ and } \forall j\in B,
\end{aligned}
\end{equation}
where the redundant constraint $K_j \leq N_U$ is added for the analysis of convergence of the following distributed algorithm, which is further explained in the proof of Theorem~\ref{theo2}. 

The only coupling constraint is $\sum_i x_{ij}=K_j$ in problem (\ref{primal}). This motivates us to turn to the Lagrangian dual decomposition method whereby a Lagrange multiplier $\mu$ is introduced to relax the coupled constraint. The dual problem is thus:
\begin{equation}\label{dual}
\textbf{D:} \quad \min\limits_\mu \quad D(\mu)=f_x(\mu)+g_K(\mu),
\end{equation}
where
\begin{equation}\label{sub1}
f(\mu)=\left\lbrace
\begin{aligned}
\max\limits_x & \quad \sum\limits_i\sum\limits_j x_{ij}\left(\log(c_{ij})-\mu_{j}\right)\\
\textrm{s.t.} & \quad \sum\limits_j x_{ij}=1\\
& \quad 0\leq x_{ij} \leq 1
\end{aligned}
\right. 
\end{equation}
\begin{equation}\label{sub2}
\begin{aligned}
g(\mu)=\max\limits_{K\leq N_U} \quad \sum\limits_j K_j\left(\mu_j -\log(K_j)\right).
\end{aligned}
\end{equation}

When the optimal value of (\ref{primal}) and (\ref{dual}) is the same, we say that strong duality holds. Slater's condition is one of the simple constraint qualifications under which strong duality holds. The constraints in (\ref{primal}) are all linear equalities and inequalities, and thus the Slater condition reduces to feasibility~\cite{BoyVan04}. Therefore, the primal problem (\ref{primal}) can be equivalently solved by the dual problem (\ref{dual}). Denoting $x_{ij}(\mu)$ as the maximizer of the first sub-problem~(\ref{sub1}) and $K_j(\mu)$ as the maximizer of the second sub-problem~(\ref{sub2}). There exits a dual optimal~$\mu^*$ such that $x(\mu^*)$ and $K(\mu^*)$ are the primal optimal. Therefore, given the dual optimal~$\mu^*$, we can get the primal optimal solution by solving the decoupled inner maximization problems (\ref{sub1}) and (\ref{sub2}) separately without coordination among the users and BSs.

\subsection{The Distributed Algorithm}
The outer problem is solved by the gradient projection method \cite{BerTsi89}, where the Lagrange multiplier~$\mu$ is updated in the opposite direction to the gradient $\nabla D(\mu)$. Evaluating the gradient of the dual objective function requires us to solve the inner maximization problem, which has been decomposed into two sub-problems~$f$ and $g$. These two sub-problems can be solved in a distributed manner.  The $t$th iteration of gradient projection algorithm is given as follows: 
\subsubsection{\textbf{User's Algorithm}}
\begin{enumerate}[i)]
\item
Each user measures the SINR by using pilot signals from all BSs, and receives the value of~$\mu_j$ broadcast by each BS at the beginning of the iteration.
\item
User $i$ determines BS $j^*$ which satisfies the follows:
\vspace{-0.4cm}
\begin{equation}\label{xij}
j^*=\arg\max\limits_j \left(\log(c_{ij})-\mu_j(t)\right).\vspace{-0.2cm}
\end{equation}
If there are multiple maximizers, user will choose any one of them.
\end{enumerate}

\subsubsection{\textbf{BSs' Algorithm}} \ 

Each BS updates the new value of $K_j$ and $\mu_j$ in two steps and announces the new multiplier $\mu_j$ to the system.
\begin{enumerate}[i)]
\item
To obtain the maximizer of problem (\ref{sub2}), we set its gradient to be 0 with the constraint $K_j \leq N_U$, i.e.,
\vspace{-0.1cm}
\begin{equation}
K_j(t+1)=\min\{N_U, e^{(\mu_j(t)-1)}\}.\vspace{-0.3cm}
\end{equation}
\item
The new value of the Lagrange multiplier is updated by
\vspace{-0.1cm}
\begin{equation}\label{muadjust}
\mu_j(t+1)=\mu_j(t)-\delta(t)\cdot\left(K_j(t)-\sum\limits_i x_{ij}(t)\right),
\end{equation}
where $\delta(t)>0$ is a dynamically chosen stepsize sequence based on some suitable estimates.
\end{enumerate}

There is a nice interpretation of $\mu$. The multiplier $\mu$ works as a message between users and BSs in the system. In fact, it can be interpreted as the price of the BSs determined by the load situation, which can be either positive or negative. If we interpret $\sum\limits_i x_{ij}$ as the serving demand for BS $j$ and~$K_j$ as the service the BS $j$ can provide, then $\mu_j$ is the bridge between demand and supply, and Eq. (\ref{muadjust}) is indeed consistent with the \textit{law of supply and demand}: if the demand $\sum\limits_i x_{ij}$ for BS $j$ exceeds the supply $K_j$, the price $\mu_j$ will go up; otherwise, the price $\mu_j$ will decrease. Thus, when the BS $j$ is over-loaded, it will increase its price $\mu_j$ and fewer users will associate with it, while other under-loaded BSs will decrease the price so as to attract more users. Moreover, the function of $\mu$ (\ref{xij}) in distributed algorithm motivates rate bias scheme, which is discussed in Section~\ref{sec:insights}.

Given $x_{ij}(\mu)$ and $K_j(\mu)$, the adjustment (\ref{muadjust}) can be made completely distributed among BSs based on only local information. At each iteration, the complexity of the distributed algorithm is $\mathcal{O}(N_BN_U)$. As for the exchanged information, at each iteration each BS broadcasts its $\mu_j$ which is a relatively small real number, and each user reports its association request to only one BS which it wants to connect to. The amount of information to be exchanged in the distributed algorithm is $k(N_B+N_U)$, where $k$ is the number of iterations, while in the centralized method it is proportional to $(N_B\times N_U)$. The gradient method converges fast generally, especially with the dynamic stepsize proposed in Sec. IV-C, and thus $k$ is a small number (less than 20 in simulation). Therefore, even with the requirement of multiple exchanges of  information, distributed algorithms may be superior for some cases, such as large scale problems. It is applicable as long as the convergence of distributed algorithm is faster than the association period. After iteratively performing the above steps, the algorithm is guaranteed to  converge to a near-optimal solution. This is proved in the next subsection.

\subsection{Step Size and Convergence}
Suppose the stepsize dynamically updates according to the rule
\begin{equation}\label{eq:stepsize}
\delta(t)=\gamma(t) \frac{D(\mu(t))-D(t)}{\| \partial D(\mu(t)) \|^2}, \quad 0<\underline{\gamma}\leq\gamma(t)\leq\bar{\gamma}<2,
\end{equation} 
where $D(t)$ is an estimate of the optimal value $D^*$ of problem (\ref{dual}), $\underline{\gamma}$ and $\bar{\gamma}$ are some scalars \cite{Ber09}. We consider a procedure for updating $D(t)$, whereby $D(t)$ is given by 
\vspace{-0.3cm}
\begin{equation}\label{eq:D-range}
D(t)=\min\limits_{0\leq\tau\leq t} D(\mu(\tau))-\varepsilon(t),
\end{equation}
and $\varepsilon(t)$ is updated according to
\vspace{-0.2cm}
\begin{equation}\label{eq:epsilonupdate}
\varepsilon(t+1)=
          \left\{
            \begin{array}{l l l}
            & \rho\varepsilon(t),  \quad & \textrm{if } D(\mu(t+1))\leq D(\mu(t)),\\
            & \max\{\beta\varepsilon(t),\varepsilon\}, \quad & \textrm{if } D(\mu(t+1))>D(\mu(t)),
            \end{array}
          \right.
\end{equation}
where $\varepsilon$, $\beta$ and $\rho$ are fixed positive constants with $\beta<1$ and $\rho>1$ \cite{Ber09}.

Thus in this procedure, we want to reach to a target level $D(t)$ that is smaller by $\varepsilon(t)$ over the best value achieved. Whenever the target level is achieved, we increase $\varepsilon(t)$ (i.e., $\rho>1$) or we keep it at the same value (i.e., $\rho=1$). If the target level is not attained at a given iteration, $\varepsilon(t)$ is reduced up to a threshold $\varepsilon$, which guarantees that the stepsize $\delta(t)$ (\ref{eq:stepsize}) is bounded away from zero. As a result, we have the following theorem.
\begin{theo}\label{theo2}
Assume that the stepsize $\delta(t)$ is updated by the dynamic stepsize rule (\ref{eq:stepsize}) with the adjustment procedure (\ref{eq:D-range}) and (\ref{eq:epsilonupdate}). If $D^*>-\infty$ where $D^*$ denotes the optimal value, then
\vspace{-0.35cm}
\begin{equation}
\inf\limits_{t} D(\mu(t))\leq D^*+\varepsilon.
\end{equation}
\end{theo}
\vspace{-0.4cm}

\noindent \begin{proof}
The derivative of function $D(\mu)$ (\ref{dual}) is given by
\begin{equation}\label{eq:partialD}
\frac{\partial D}{\partial \mu_j}(\mu)=K_j(\mu)-\sum\limits_i x_{ij}(\mu).
\end{equation}
In our primal problem, $K_j=\sum_i x_{ij}\leq N_U$ where $N_U$ is the total number of users. According to (\ref{eq:partialD}), when $K_j$ and $\sum_i x_{ij}$ are bounded, the subgradient of dual objective function $\partial D$ is also bounded:
\begin{equation}\label{eq:subgradbound}
\sup\limits_{t} \{ \| \partial D(\mu(t)) \| \} \leq c,
\end{equation}
where $c$ is some scalar. Thus, our problem satisfies the necessary conditions of Proposition~6.3.6 in  \cite{Ber09}. By applying this proposition, the theorem is proved.
\end{proof}

\section{Range Expansion (Biasing)}\label{sec:insights}
The proposed approaches above is sensitive to the deployment of users and BSs, i.e., the algorithms have to run again and again in order to keep tracking of changes in networks. In this section, we investigate a simple approach, called range expansion, which is insensitive to the change of deployments. Range expansion is proposed a practical way to balance loads in HetNets, since it allows for a simple uncoordinated decision based only on the received power from a given BS \cite{Bje11,DamMon11}. It is implemented by assigning a multiplicative SINR bias to each tier of BSs (depending primarily on their transmit power).  For example, if a picocell had a 10 dB SINR bias vs. the macrocell BS, a user would associate with it until the SINR delivered by the macro BS was a full 10 dB higher than the picocell.  This can be performed by measuring the pilot signals from the BSs within radio range and then simply associating with the one that has the highest \emph{biased} received power.  In this section we investigate whether this simple approach is compatible with the optimal load-aware problems formulated in the prior two sections.

There are some recent studies on the SINR bias \cite{JoSan11,SalBul10}, but have not given any theoretical guidance on the ``best'' biasing factors in the sense of load balancing and/or achieving some optimization criteria.  In this section, we evaluate the range expansion that our optimal user association scheme provides in terms of SINR bias. Moreover, the distributed algorithm inspires a \emph{rate} bias scheme where the biasing factor is multiplied with the rate instead of SINR. The best SINR biasing factor is obtained by a brute force search based on the optimal FUA, and the best rate biasing factor is derived directly from the optimal $\mu_j^*$ in the dual distributed algorithm. The network-wide performance with either biasing factor gets pretty close to the optimal FUA, among which rate bias performs better than SINR bias. A more interesting observation is that the biasing factors are \emph{insensitive to the location of BSs and users, which makes the bias schemes simple and robust to implement in practice.}

\subsection{SINR Bias}
We first consider the SINR bias, where users are associated with the BS which provides the highest \textit{biased} SINR.
\begin{defi}
Given the biasing factor $A_j$ for BS $j$, we define the biased SINR received by user $i$ from BS $j$ as
\begin{equation}\label{eq:biassinr}
\textrm{SINR}_{ij}'=A_j\cdot \textrm{SINR}_{ij}= \frac{A_j\cdot P_jg_{ij}}{\sum\limits_{k\in\mathcal{B},k\neq j} P_k g_{ik}+N_0}.
\end{equation}
\end{defi}

We adopt an identical biasing factor for all BSs in the same tier \cite{JoSan11,SalBul10,DamMon11}. Note that setting the biasing factors at all tiers to 1 reduces to the conventional max-SINR cell association, and setting them to~$A_j=\left. 1 \middle/ P_j \right.$ associates users to the BS with the lowest path loss. Biasing under-loaded small BSs, the cells extend the coverage and attract more users, thus resulting in a more fair distribution of traffic. In simulation, the biasing factor is quite stable as the change of BS density and transmit power, and the performance by SINR bias is very close to the optimal performance by FUA, which is further discussed in Section \ref{sec:simulation}.

\subsection{Rate Bias}
According to our load aware association schemes, the best SINR biasing factors are obtained by a brute force search with high complexity. The solution (\ref{xij}) of the dual distributed algorithm motivates the more tractable idea of rate bias. According to (\ref{xij}), user $i$ is associated with BS~$j^*$, where $j^*=\arg\max\limits_j (c_{ij}e^{-\mu_j^*})$ for optimal association. Therefore, by setting the rate biasing factor to $B_j=e^{-\mu_j^*}$, the association would be exactly same as the association obtained by the distributed algorithm, which is a near-optimal solution.

\begin{defi}
We define the biased rate $c_{ij}$ of user $i$ from BS $j$ as
\vspace{-0.4cm}
\begin{equation}\label{eq:biasc}
c'_{ij}=c_{ij}\cdot B_j. \vspace{-0.3cm}
\end{equation}
\end{defi}
\noindent Through range expansion, users are associated with the BS that serves the maximum biased rate~$c'_{ij}$. In rate bias, the biasing factor is in the exponential term of SINR (i.e., $(1+\textrm{SINR}_{ij})^{B_j}$), which is different from SINR bias where the biasing factor is multiplied directly to SINR (i.e., $A_j\textrm{SINR}_{ij}$). 

In the distributed algorithm, the price variable is different from BS to BS, even for those belonging to the same tier. However, in the investigation of range expansion, just as for SINR bias, we use the same biasing factor for all BSs in a given tier, which is the mean of the optimal multiplier, i.e., $B_j=E[e^{-\mu_l^*}]$, where~$l\in j$th tier. The results of rate bias shown in next section is very close to optimal solution in FUA.
\section{Performance Evaluation}\label{sec:simulation}
We consider a three-tier HetNet with transmit power $\{P_1, P_2, P_3\}=\{46,35,20\}$dBm. The theoretical analysis throughout this paper is independent of the spatial distribution of the BSs. For the simulations, we model the locations of the macro BSs to be fixed, and the locations of the small BSs to be uniformly and independently distributed in space. This corresponds to operator deployed macros/picocells, and customer-placed femtocells. We model the location processes across different tiers as independent, with deployed density $\{\lambda_2, \lambda_3\}=\{ 5,20\}$ per macrocell. In modelling the propagation environment, we use a path loss $L(d) = 34 + 40 \log(d)$ and $L(d) = 37 + 30 \log(d)$ for macros/picocells and femtocells respectively. We assume lognormal shadowing with a standard deviation $\sigma_s=8$dB. At room temperature and bandwidth 10MHz, the thermal noise power is $\sigma^2=kTB=-104$dBm. We then assume that during the connection period between user $i$ and BS $j$, the user achieves the Shannon capacity rate, i.e., $c_{ij}=\log_2(1+\textrm{SINR}_{ij})$.

\subsection{Loads among different BSs}
Fig. \ref{com-mean} compares the load situations among different association schemes. The max-SINR association results in very unbalanced loads: the macro BSs are over-loaded, while small BSs serve far fewer users, with some even being idle. In the fractional association scheme, the load is shifted to the less congested small BSs, which suggests that our objective alleviates the asymmetric load problem. The results after rounding are almost the same as the global optimum obtained by fractional association, showing the effectiveness of the rounding scheme. This occurs because there are few users associated with more than one BSs: most users are not ``fractional''. Moreover, the fractional users  usually have a strong preference towards one of the BSs. The dual distributed algorithm and biasing also provide near-optimal load distributions with low complexity.

\begin{figure}[h]\centering
\includegraphics[width=14cm, height=8cm]{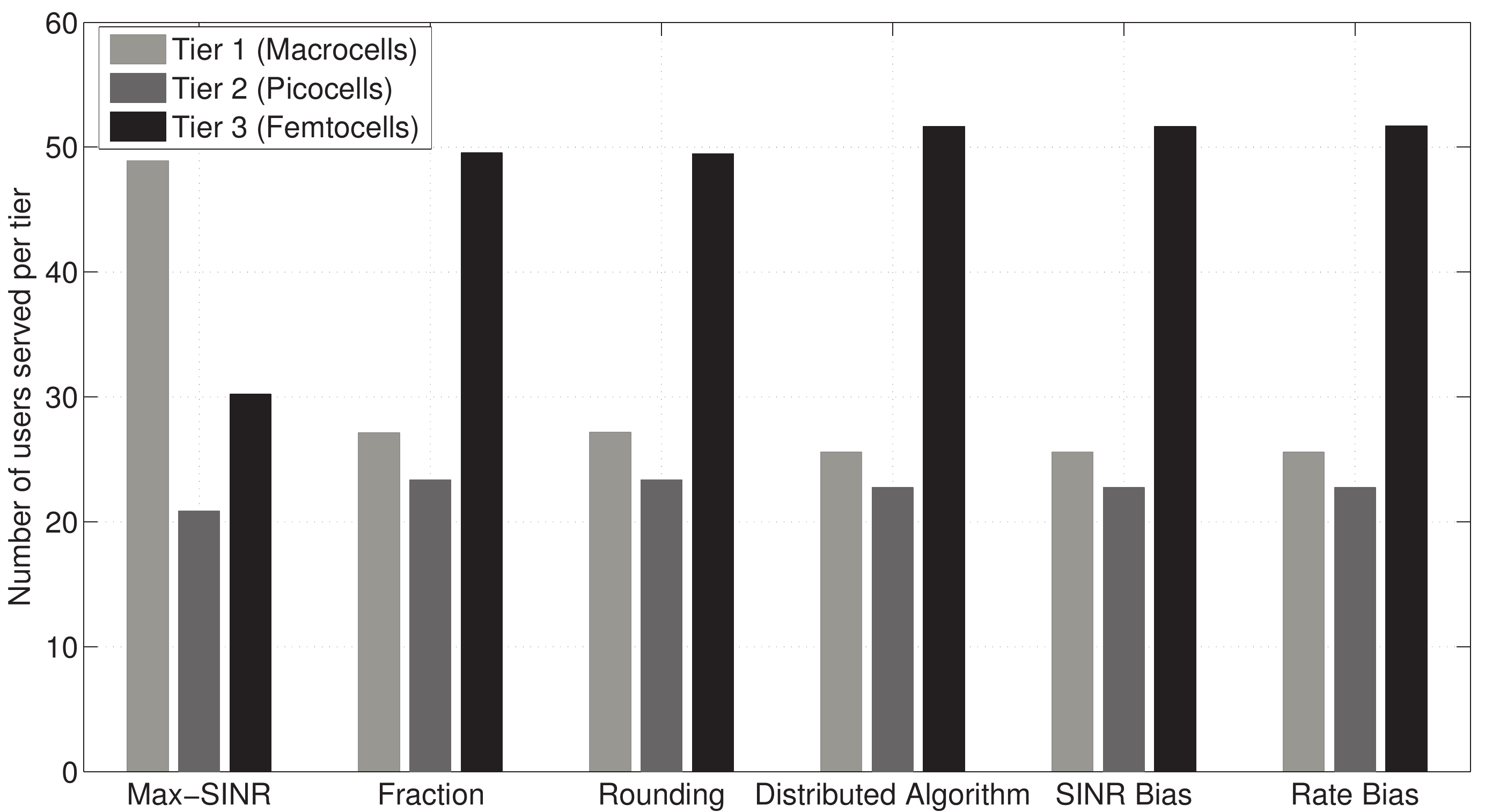}
\small
\caption{Comparisons of average number of users per tier in a three-tier HetNets. Compared to max-SINR, the load is shifted from over-loaded Macro BSs to the less congested small BSs in fractional association. Rounding method and biasing schemes obtain very close results to optimal results in fractional association.}
\label{com-mean}
\end{figure}

\subsection{Rate CDF}
As another performance measure, Fig. \ref{fig:cdf} and \ref{fig:cdfmacro} show the cumulative distribution function (CDF) of long-term rate in HetNets and conventional networks with different association schemes respectively. In HetNets, the CDFs for joint association, fraction-rounding, the dual distributed algorithm and biasing all improve significantly at low rate vs. max-SINR, showing a 2-3.5x gain, in both static setting and stochastic setting.The CDF of max-SINR catches up at a rate of $0.3$ bits/s/Hz, since load balancing provides a more uniform user experience by taking resources from strong users. However, load balancing enables the system to accommodate more users, which will boost the system revenue in general. The CDFs of fractional-rounding and the dual distributed algorithm almost overlap, which verifies that the distributed algorithm converges to a near-optimal solution. The result of joint association is very close to the result of fraction-rounding, which verifies Prop. \ref{prop:equalresource}. Note that in stochastic setting, we adopt PF as the scheduling scheme. The static channel equals the average of stochastic channel. From fig. \ref{fig:cdf}, we can see that the rate in stochastic setting with PF is larger than the rate in static setting, although by PF, the resource allocation will eventually converge to almost equal allocation for each user. This is because the channel distribution would be changed by PF (users are more possible to be served in good channel status, i.e., the cij would be larger than the average rate defined in this paper). Fig. \ref{fig:cdfmacro} shows that the rate gain is unique for HetNets as long as the users are uniformly distributed. 
\begin{figure}[h]\centering
\includegraphics[width=14cm, height=8cm]{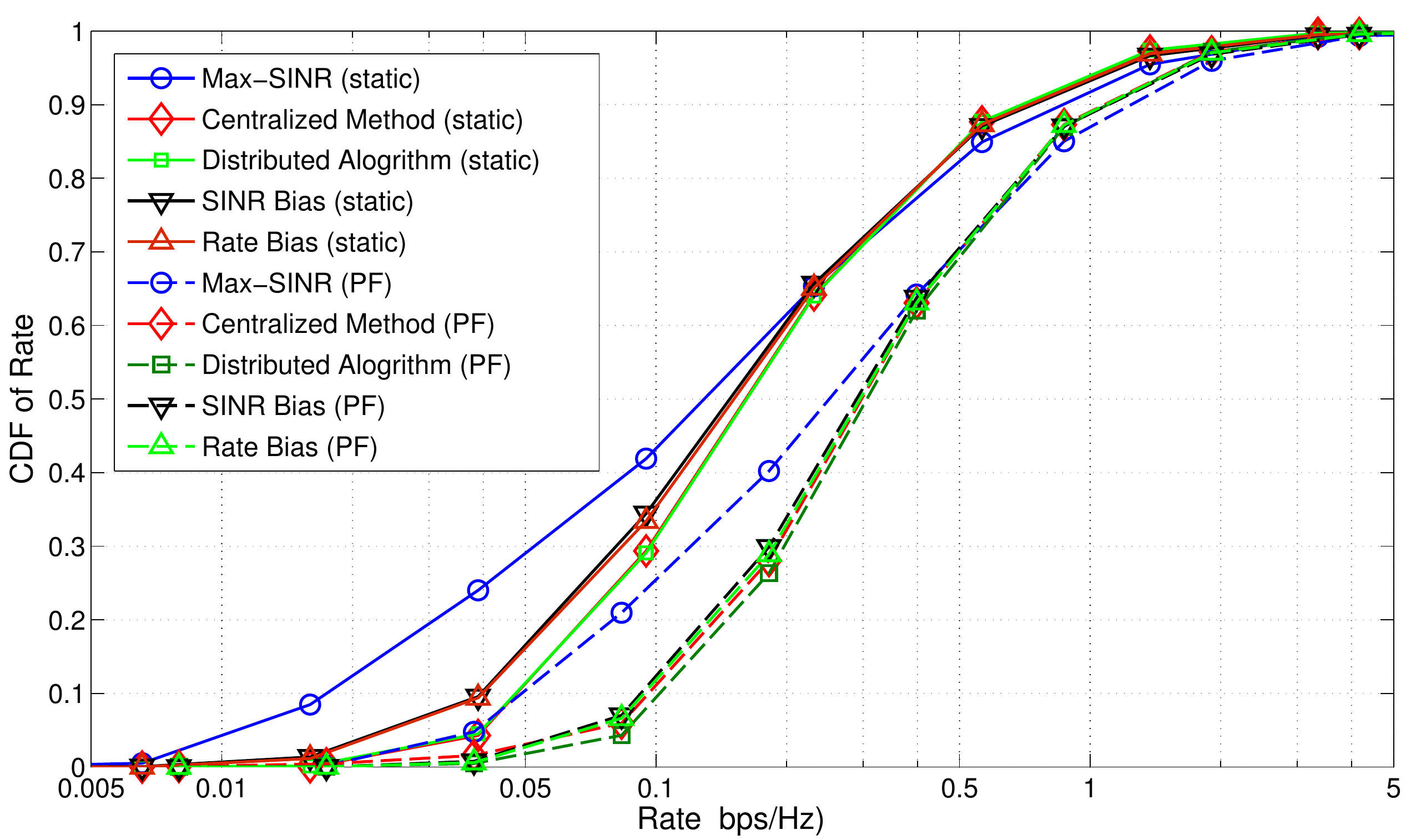}
\small
\caption{The CDFs of overall rate in a three-tier HetNets, in both static setting and stochastic setting. The biasing factors of macro BSs, picos and femtos are $\{A_1, A_2, A_3\}=\{1.00, 4.00, 11.9\}$ in SINR bias, and $\{B_1, B_2, B_3\}=\{1.00, 1.59, 1.88\}$ in rate bias, respectively.}
\label{fig:cdf}
\end{figure}
\begin{figure}[h]\centering
\includegraphics[width=14cm, height=8cm]{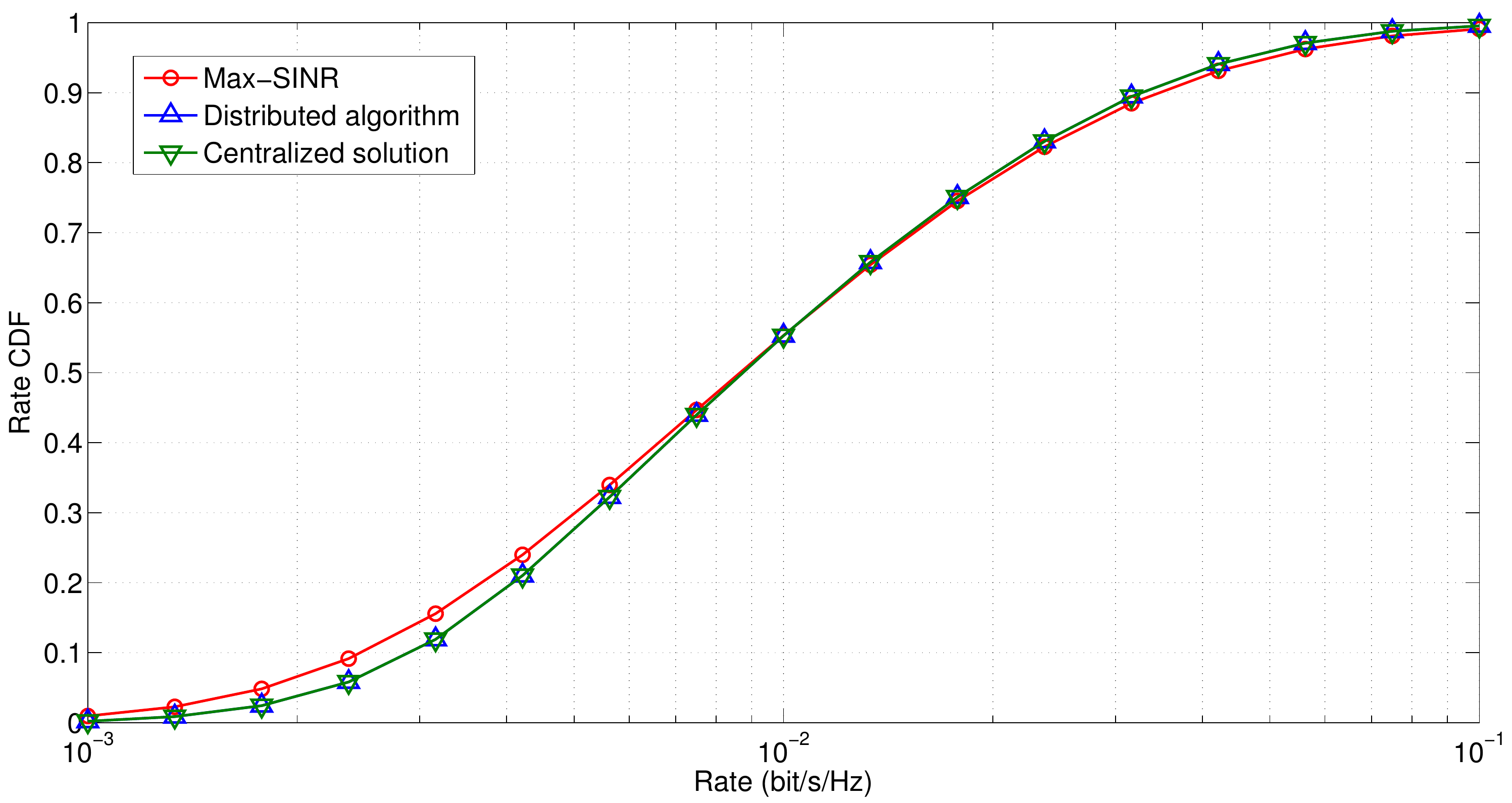}
\small
\caption{The CDFs of overall rate in macro-only networks.}
\label{fig:cdfmacro}
\end{figure}

The ratio of rate $\alpha$ vs. probability $\mathbb{P}(R<\alpha)$ for the various approaches vs. max-SINR is represented in Fig. \ref{fig:rateratio}. The rate gain is quite large (e.g., 3.5x vs. max-SINR at the 10\% rate point). The results for simple biasing are very close to the optimum associations, where the empirically observed biasing factors of macrocells, picocells and femtocells turned out to be $\{A_1, A_2, A_3\}=\{0, 6, 10.8\}$ dB in SINR, and $\{B_1, B_2, B_3\}=\{1.00, 1.59, 1.88\}$ in rate (linear units) respectively, for the chosen parameters.
\begin{figure}[!hbp]
\centering
\includegraphics[width=14cm, height=8cm]{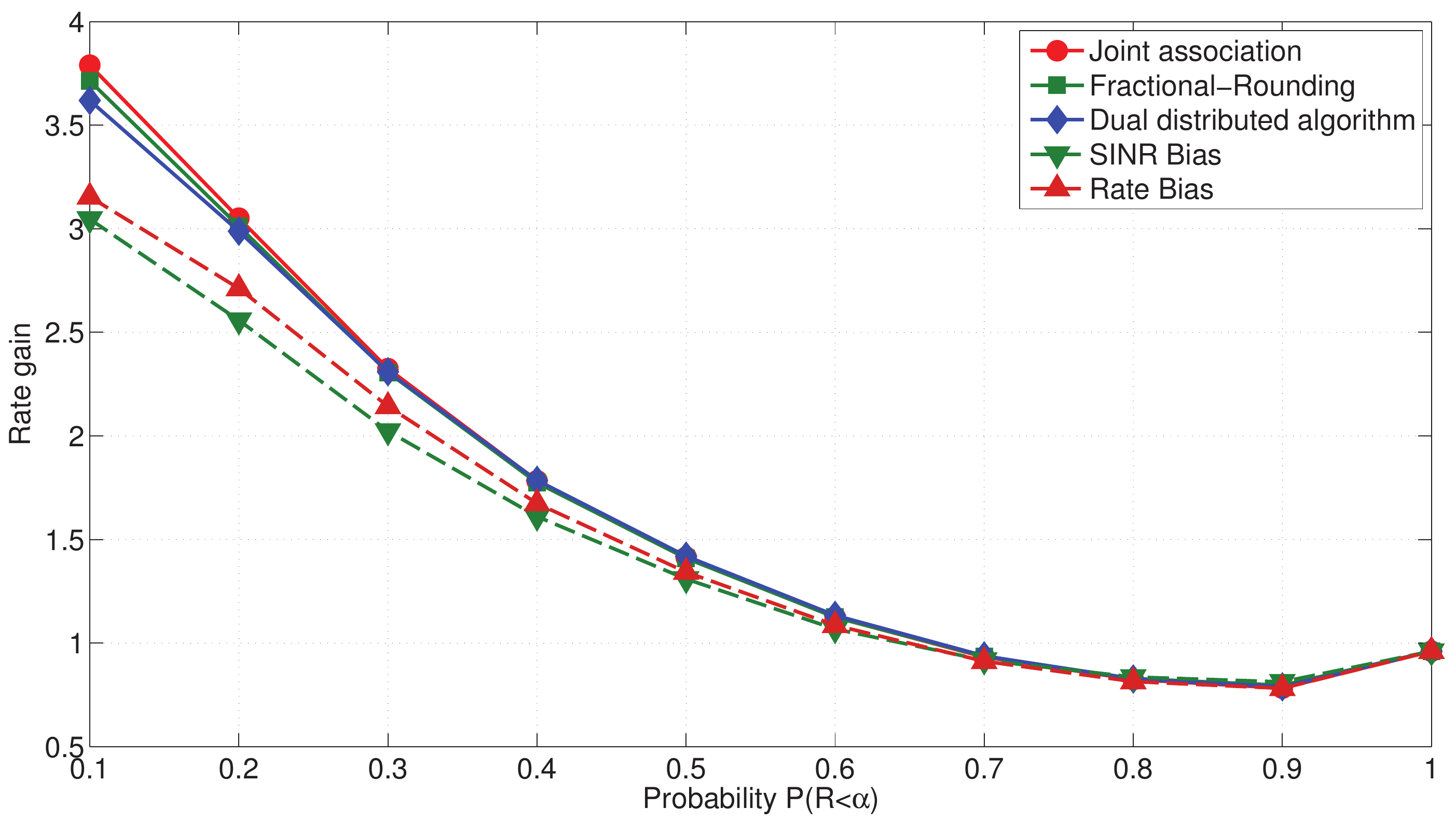}
\caption{Rate gain in a three-tier HetNet. The rate ratio of joint association scheme, fractional-rounding scheme, dual distributed algorithm, and biasing schemes to max-SINR is represented. There is a very large gain for the bottom half of users, i.e., cell edge users.}
\label{fig:rateratio}
\end{figure}
\subsection{Biasing Factor}
The effect of BS density and transmit power on biasing factors is considered in Fig. \ref{fig:bfdensity} and \ref{fig:bfpower} respectively. The biasing factors have been normalized, which means no biasing for the macrocell tier.

When the deployed density of small BSs changes, it is interesting to observe in Fig. \ref{fig:bfdensity2} and \ref{fig:bfdensity3} that deploying more small BSs has very little effect on the biasing factor. Intuitively, though the density of BSs increases, within a reasonable change range of density, there are more users associated with that type of BSs in the optimal association, which makes the needed range expansion almost the same as that in the original scenario. Therefore, \emph{the optimal biasing factors will be almost the same as the network infrastructure deployment evolves}.
\begin{figure}\centering
   \renewcommand{\subcapsize}{\tiny}
   \setcounter{subfigure}{0}
   \setlength{\abovecaptionskip}{0pt}

   \subfigure{\label{fig:bfdensity2}%
   \includegraphics[width=8.5cm, height=6cm]{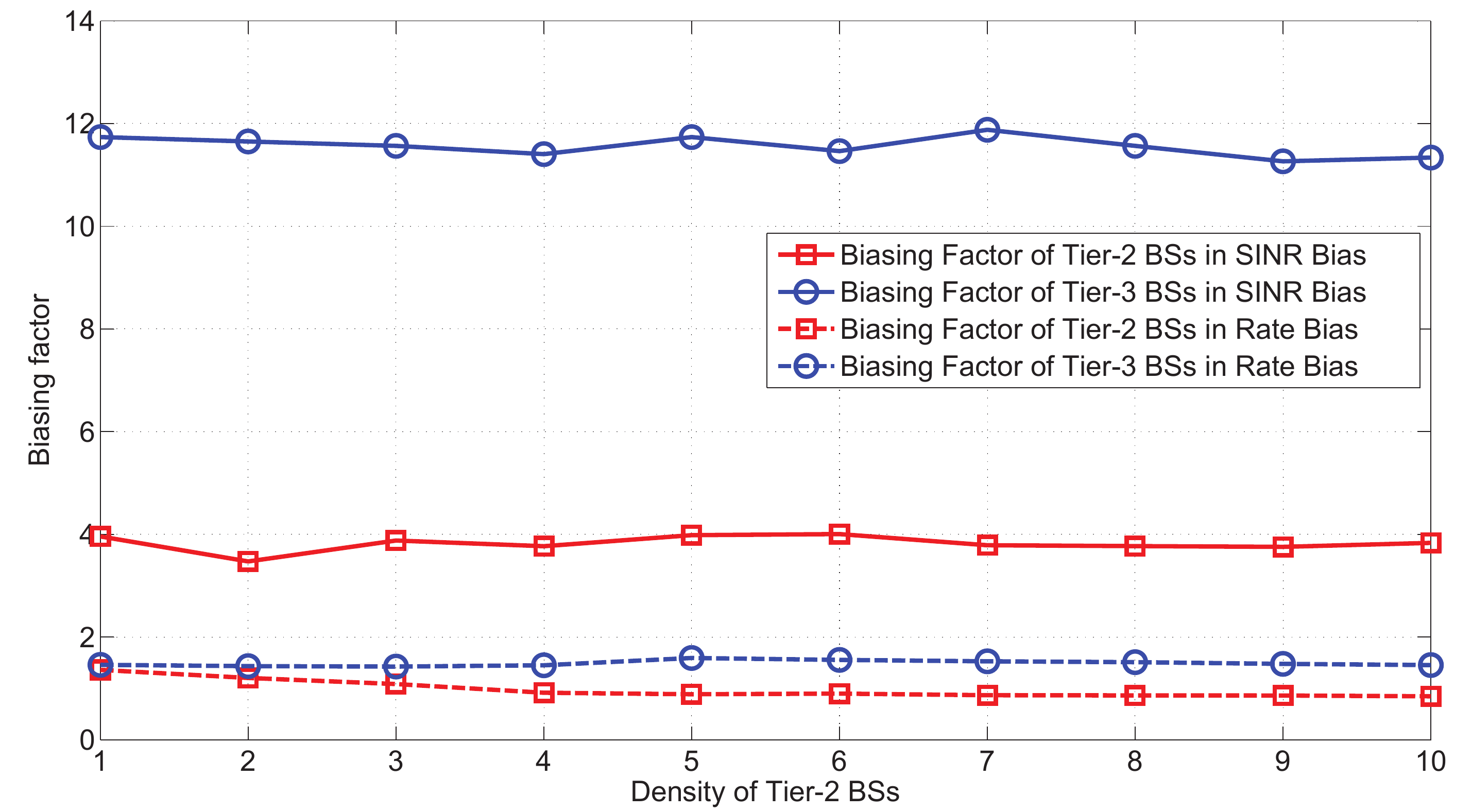}}
   \subfigure{\label{fig:bfdensity3}
   \includegraphics[width=8.5cm, height=6.2cm]{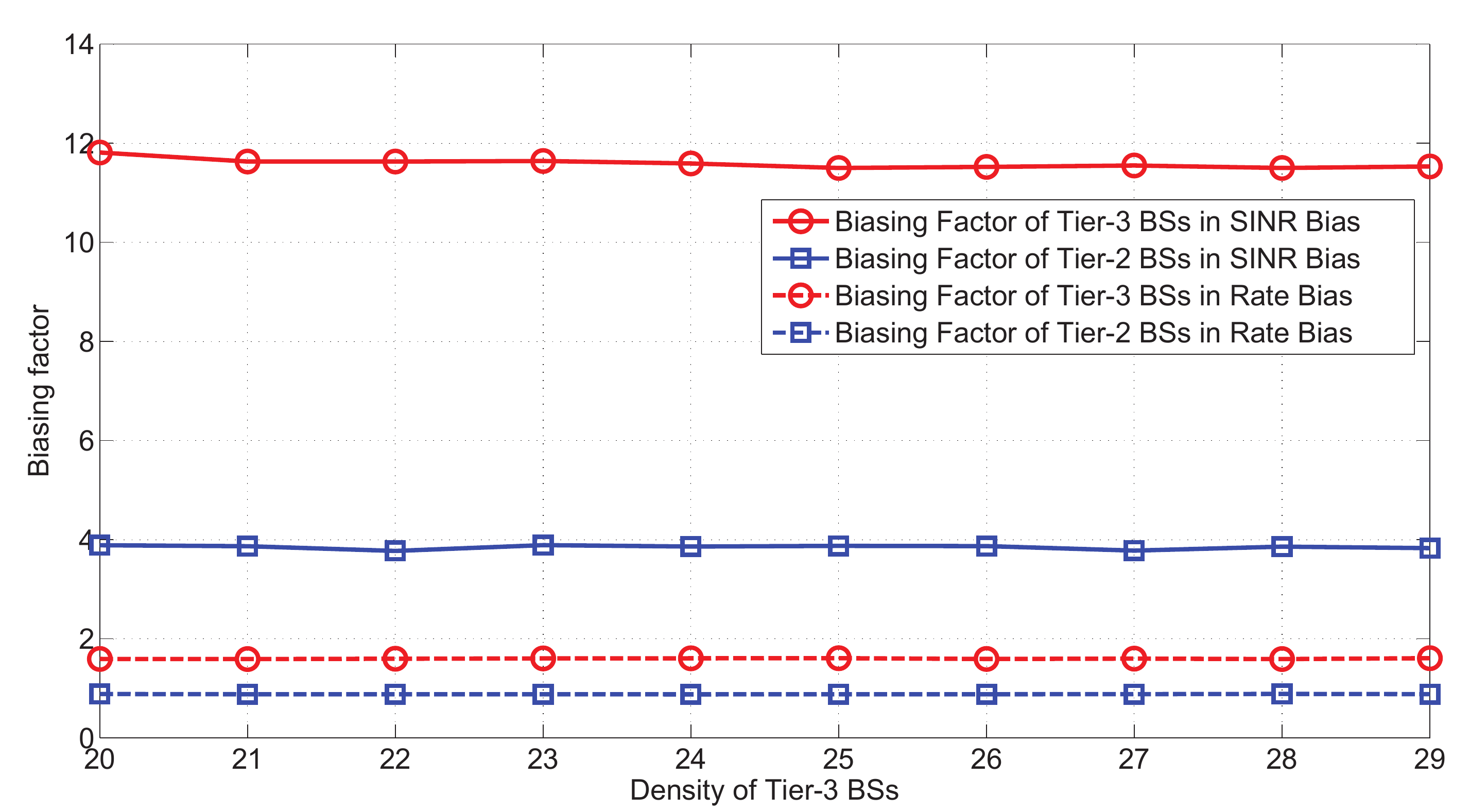}}
\caption{Biasing factors vs. density of small BSs in a three-tier HetNet. The density of one tier changes while the others are fixed, with tier 1 always having biasing factor 1. Deploying more small BSs has almost no effect on biasing.}\label{fig:bfdensity}
\end{figure}

However, the story is quite different when the transmit power changes. As power of 2nd-tier BSs increase in Fig \ref{fig:bfpower2}, the biasing factor of 2nd-tier BSs steadily decreases, while the biasing factor of 3rd-tier BSs almost stays the same. A similar conclusion can be obtained from Fig. \ref{fig:bfpower3}, where the biasing factor of 3rd-tier BSs decreases gradually and the biasing factor of 2nd-tier BSs is almost static. The biasing factor is smaller as the transmit power increases because users are more likely to be associated with these BSs using max-biased SINR even without a strong bias.

\begin{figure}\centering
   \renewcommand{\subcapsize}{\tiny}
   \setcounter{subfigure}{0}
   \setlength{\abovecaptionskip}{0pt}

   \subfigure{\label{fig:bfpower2}%
   \includegraphics[width=8.5cm, height=6cm]{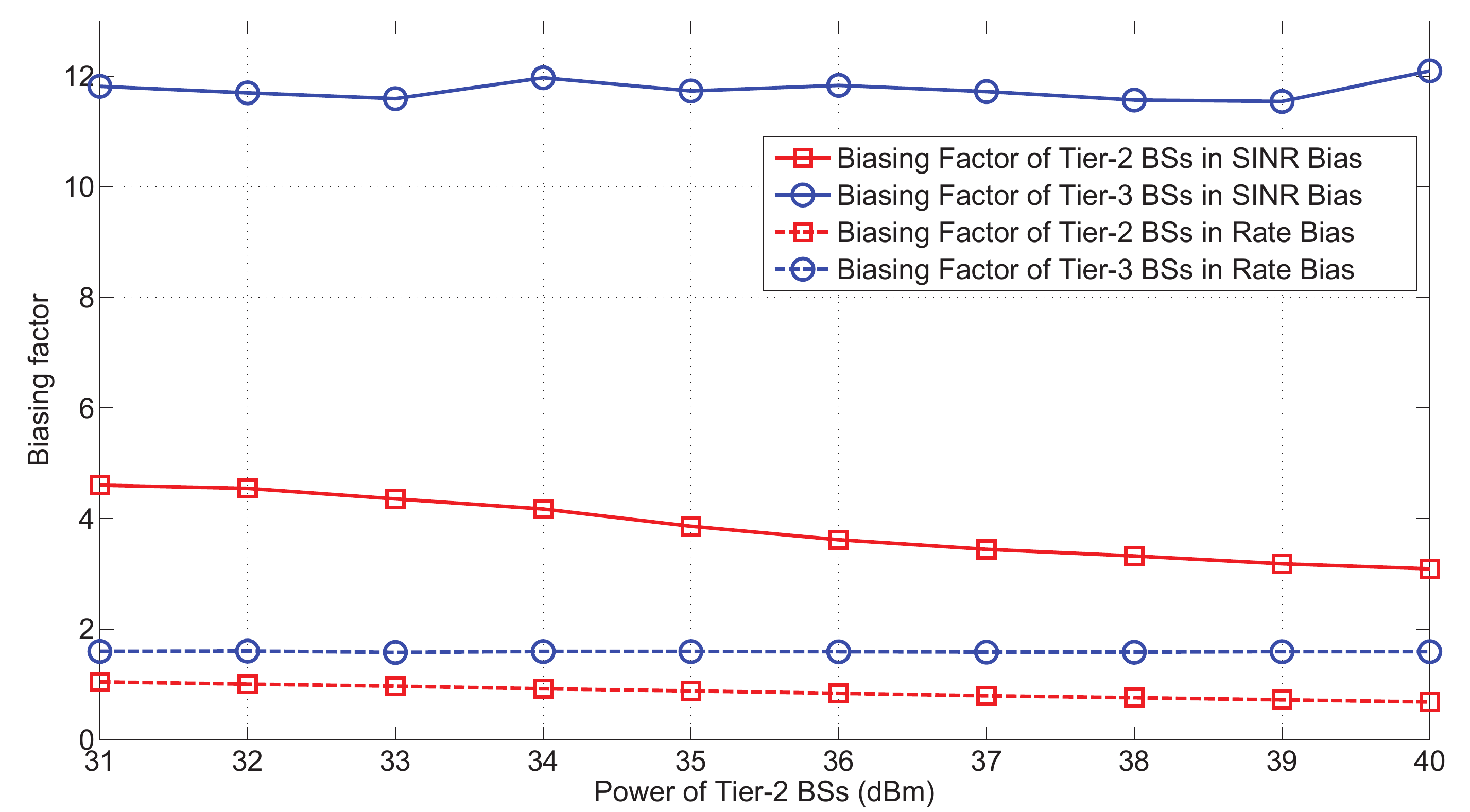}}
   \subfigure{\label{fig:bfpower3}
   \includegraphics[width=8.5cm, height=6.15cm]{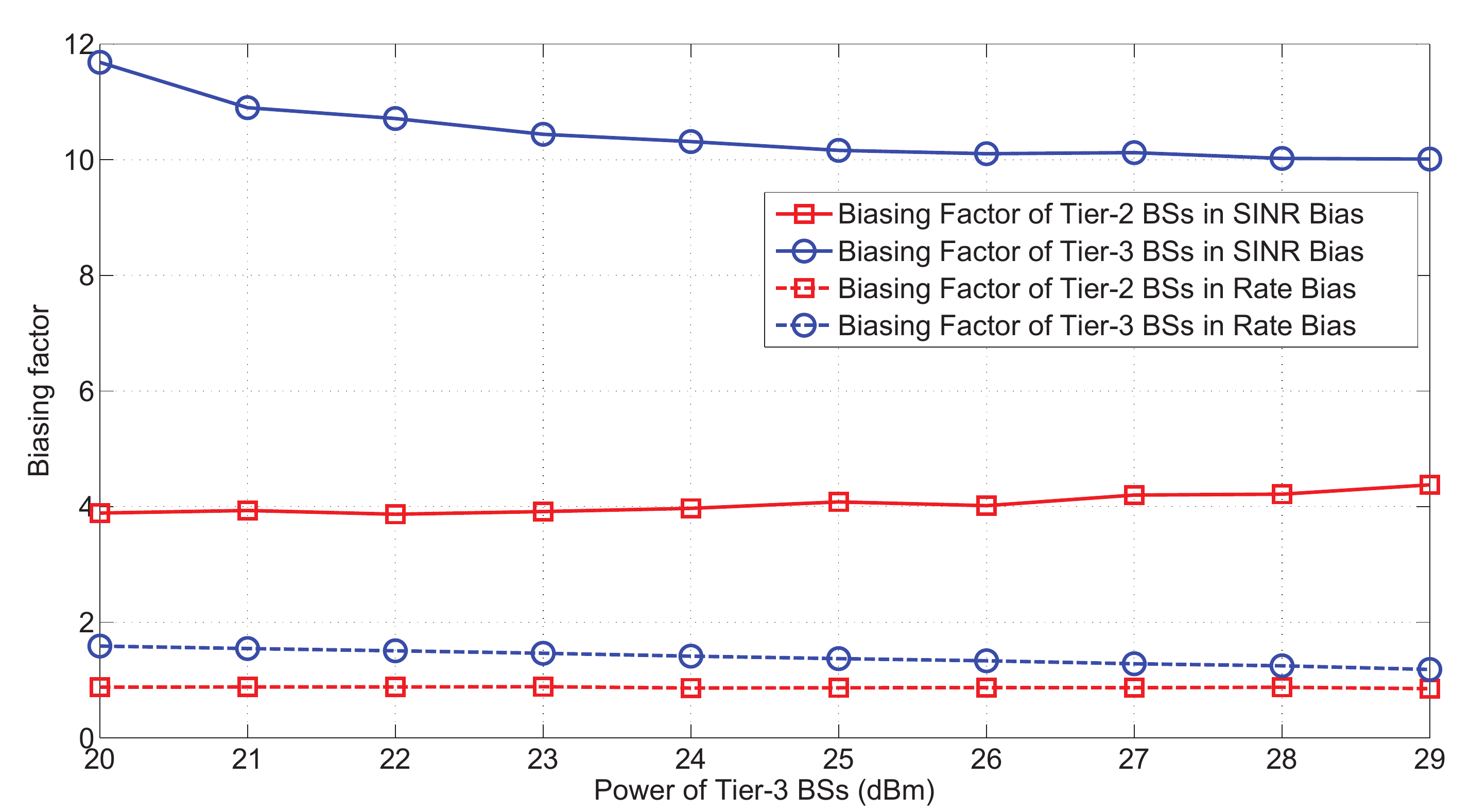}}
\caption{Biasing factors vs. transmit power of BSs in a three-tier HetNet, with tier 1 having biasing factor 1.}\label{fig:bfpower}
\end{figure}

\section{Conclusion}\label{sec:conclusion}
In this paper, we propose a class of novel user association schemes that achieve load balancing in HetNets through a network-wide utility maximization problem. We first consider the cell association and resource allocation jointly, and propose an upper bound on performance. Then we formulate a logarithmic utility maximization problem where the equal resource allocation is optimal, and design a distributed algorithm via dual decomposition, from the relaxation of physical constraints. The distributed algorithm is proved to converge to a near-optimal solution, with low complexity that is linear to the number of users and the number of BSs. Finally, our scheme is extended to the range expansion technique, which requires limited changes to the existing system architecture by introducing biasing factors to small BSs. We consider two types of biasing factors (SINR and rate), and evaluate the effects of BSs' density and transmit power on the biasing factors by using our load-aware association scheme. 

A key observation is that the optimal biasing factors are nearly independent of the BS densities for the various tiers, but highly dependent on the per-tier transmit powers. With these optimal biasing factors, the network nearly achieves the optimal load-aware performance.  The numerical results demonstrate that a load-aware association significantly improves resource utilization and mitigates the congestion of macro BSs, resulting in a multi-fold gain to the overall rate for most users, particularly those with previously low rates. Future work could include load-aware user associations that incorporate dynamic settings (dynamic traffic, high user mobility),  include providing a analytical framework of biasing factors and finding the optimal value theoretically, include the uplink scenario with power control, and the consideration of additional utility functions.


\appendices
\section{Proof of Proposition \ref{prop:equalresource}}\label{proof:equalresource}
The objective function of (\ref{eq:jointop}) is
\vspace{-0.5cm}
\begin{equation}\label{eq:jointop2}
\max_y  \sum_{i\in\{l| x_{lj}=1\}} \log(y_{ij}c_{ij})= \sum_{i\in\{l| x_{lj}=1\}} \log(c_{ij})+\log(y_{ij}),\vspace{-0.3cm}
\end{equation}

where $\sum_i\log(c_{ij})$ is constant relative to $\textrm{SINR}_{ij}$. Then the objective function is equivalent to maximizing the geometric mean:
\vspace{-0.6cm}
\begin{equation}
\max\limits_y  \sum_{i\in\{l| x_{lj}=1\}} \log(y_{ij}) \Leftrightarrow \max\limits_y \frac{1}{N_u} \log\left(\prod\limits_i^{N_u} y_{ij}\right) \Leftrightarrow \max\limits_y\sqrt[n]{y_{1j}y_{2j}\cdots y_{N_uj}},\vspace{-0.3cm}
\end{equation}
where $N_U$ denotes the number of users associated with BS $j$.
As the geometric mean is no greater than the arithmetic mean, we have
$
\sqrt[n]{y_{1j} y_{2j} \cdots y_{N_uj}} \leq \frac{y_{1j}+y_{2j}+\cdots+y_{N_uj}}{N_u},
$
where the equality holds if and only if $y_{1j}=y_{2j}=\cdots=y_{N_uj}$. Therefore, to maximize (\ref{eq:jointop2}), $y_{ij}$ should be equal for all $i$, i.e., $y_{ij}=\left.1\middle/{K_j}\right.$.

\bibliographystyle{ieeetr}
\bibliography{qyebib}

\begin{thebibliography}{10}

\bibitem{DhiAnd12}
H.~S. Dhillon, R.~K. Ganti, F.~Baccelli, and J.~G. Andrews, ``Modeling and
  analysis of {K}-tier downlink heterogeneous cellular networks,'' {\em IEEE
  Journal on Sel. Areas in Communications}, vol.~30, pp.~550--560, Apr. 2012.

\bibitem{Lag97}
X.~Lagrange, ``Multitier cell design,'' {\em IEEE Comm. Mag.}, vol.~35,
  pp.~60--64, Aug. 1997.

\bibitem{SalRus87}
A.~Saleh, A.~Rustako, and R.~Roman, ``Distributed antennas for indoor radio
  communications,'' {\em IEEE Trans. on Communications}, vol.~35,
  pp.~1245--1251, Dec. 1987.

\bibitem{SydTao09}
J.~Sydir and R.~Taori, ``An evolved cellular system architecture incorporating
  relay stations,'' {\em IEEE Comm. Mag.}, vol.~47, pp.~115--121, June 2009.

\bibitem{WuMur04}
X.~Wu, B.~Murherjee, and D.~Ghosal, ``Hierarchical architectures in the
  third-generation cellular network,'' {\em IEEE Wireless Communications},
  vol.~11, pp.~62--71, June 2004.

\bibitem{AndCla12}
J.~G. Andrews, H.~Claussen, M.~Dohler, S.~Rangan, and M.~C. Reed, ``Femtocells:
  Past, present, and future,'' {\em IEEE Journal on Sel. Areas in
  Communications}, vol.~30, pp.~497--508, Apr. 2012.

\bibitem{DamMon11}
A.~Damnjanovic, J.~Montojo, Y.~Wei, T.~Ji, T.~Luo, M.~Vajapeyam, T.~Yoo,
  O.~Song, and D.~Malladi, ``A survey on {3GPP} heterogeneous networks,'' {\em
  IEEE Wireless Communications Magazine}, vol.~18, pp.~10--21, June 2011.

\bibitem{KahGeo78}
T.~Kahwa and N.~Georganas, ``A hybrid channel assignment scheme in large-scale,
  cellular-structured mobile communication systems,'' {\em IEEE Trans. on
  Communications}, vol.~26, pp.~432--438, Apr. 1978.

\bibitem{JiaRap94}
H.~Jiang and S.~Rappaport, ``{CBWL}: A new channel assignment and sharing
  method for cellular communication systems,'' {\em IEEE Trans. on Veh.
  Technology}, vol.~43, pp.~313--322, May 1994.

\bibitem{DasSen97}
S.~K. Das, S.~K. Sen, and R.~Jayaram, ``A dynamic load balancing strategy for
  channel assignment using selective borrowing in cellular mobile
  environment,'' {\em Wireless Networks}, vol.~3, pp.~333--347, Oct. 1997.

\bibitem{DasSen98}
S.~K. Das, S.~K. Sen, and R.~Jayaram, ``A novel load balancing scheme for the
  tele-traffic hot spot problem in cellular networks,'' {\em Wireless
  Networks}, vol.~4, pp.~325--340, July 1998.

\bibitem{Ekl86}
B.~Eklundh, ``Channel utilization and blocking probability in a cellular mobile
  telephone system with directed retry,'' {\em IEEE Trans. on Communications},
  vol.~34, pp.~329--337, Apr. 1986.

\bibitem{WuMuk00}
X.~Wu, B.~Mukherjee, and S.~H.~G. Chan, ``{MACA}-an efficient channel
  allocation scheme in cellular networks,'' in {\em Proc., IEEE Globecom},
  vol.~3, pp.~1385--1389, 2000.

\bibitem{CavAgr05}
D.~Cavalcanti, D.~Agrawal, C.~Cordeiro, B.~Xie, and A.~Kumar, ``Issues in
  integrating cellular networks {WLANs}, and {MANET}s: a futuristic
  heterogeneous wireless network,'' {\em IEEE Wireless Communications
  Magazine}, vol.~12, pp.~30--41, June 2005.

\bibitem{YanTon04}
E.~Yanmaz and O.~K. Tonguz, ``Dynamic load balancing and sharing performance of
  integrated wireless networks,'' {\em IEEE Journal on Sel. Areas in
  Communications}, vol.~22, pp.~862--872, June 2004.

\bibitem{DasVis03}
S.~Das, H.~Viswanathan, and G.~Rittenhouse, ``Dynamic load balancing through
  coordinated scheduling in packet data systems,'' in {\em Proc., IEEE
  INFOCOM}, vol.~1, pp.~786--796, Apr. 2003.

\bibitem{BejHan09}
Y.~Bejerano and S.~J. Han, ``Cell breathing techniques for load balancing in
  wireless {LAN}s,'' {\em IEEE Transactions on Mobile Computing}, vol.~8,
  pp.~735--749, June 2009.

\bibitem{SanWan08}
A.~Sang, X.~Wang, M.~Madihian, and R.~D. Gitlin, ``Coordinated load balancing,
  handoff/cell-site selection, and scheduling in multi-cell packet data
  systems,'' {\em Wireless Networks}, vol.~14, pp.~103--120, Jan. 2008.

\bibitem{BuLi06}
T.~Bu, L.~Li, and R.~Ramjee, ``Generalized proportional fair scheduling in
  third generation wireless data networks,'' in {\em Proc., IEEE INFOCOM},
  pp.~1--12, Apr. 2006.

\bibitem{BejHan07}
Y.~Bejerano, S.~J. Han, and L.~Li, ``Fairness and load balancing in wireless
  {LAN}s using association control,'' {\em IEEE/ACM Trans. on Networking},
  vol.~15, pp.~560--573, June 2007.

\bibitem{SonCho09}
K.~Son, S.~Chong, and G.~Veciana, ``Dynamic association for load balancing and
  interference avoidance in multi-cell networks,'' {\em IEEE Trans. on Wireless
  Communications}, vol.~8, pp.~3566--3576, July 2009.

\bibitem{KimVec11}
H.~Kim, G.~de~Veciana, X.~Yang, and M.~Venkatachalam, ``Distributed
  $\alpha$-optimal user association and cell load balancing in wireless
  networks,'' {\em IEEE/ACM Trans. on Networking}, pp.~1--14, June 2011.

\bibitem{CheBac11}
C.~Chen, F.~Baccelli, and L.~Roullet, ``Joint optimization of radio resources
  in small and macro cell networks,'' in {\em Proc., IEEE Veh. Technology
  Conf.}, pp.~1--5, IEEE, May 2011.

\bibitem{CorFal12}
S.~Corroy, L.~Falconetti, and R.~Mathar, ``Dynamic cell association for
  downlink sum rate maximization in multi-cell heterogeneous networks,'' in
  {\em {IEEE} International Conference on Communications ({ICC})}, to be
  published after Apr. 2012.

\bibitem{JoSan11}
H.~Jo, Y.~Sang, P.~Xia, and J.~Andrews, ``Heterogeneous cellular networks with
  flexible cell association: a comprehensive downlink {SINR} analysis,'' {\em
  submitted to IEEE Trans. on Wireless Communications}, July 2011, available at
  arxiv.org/abs/1107.3602.

\bibitem{StaWic09}
S.~Stanczak, M.~Wiczanowski, and H.~Boche, {\em {F}undamentals of {R}esource
  {A}llocation in {W}ireless {N}etworks: {T}heory and {A}lgorithms}, vol.~3.
\newblock Springer Verlag, 2009.

\bibitem{LowLap99}
S.~Low and D.~Lapsley, ``Optimization flow control{-I}: basic algorithm and
  convergence,'' {\em IEEE/ACM Trans. on Networking}, vol.~7, pp.~861--874,
  Dec. 1999.

\bibitem{BoyVan04}
S.~Boyd and L.~Vandenberghe, {\em {C}onvex {O}ptimization}.
\newblock Cambridge Univ Pr, 2004.

\bibitem{BerTsi89}
D.~Bertsekas and J.~Tsitsiklis, {\em {P}arallel and {D}istributed
  {C}omputation: {N}umerical {M}ethods}.
\newblock Prentice-Hall, Inc., 1989.

\bibitem{Ber09}
D.~Bertsekas, {\em Convex Optimization Theory}.
\newblock Athena Scientific, 2009.

\bibitem{Bje11}
B.~Bjerke, ``{LTE}-advanced and the evolution of {LTE} deployments,'' {\em IEEE
  Wireless Communications}, vol.~18, pp.~4--5, Oct. 2011.

\bibitem{SalBul10}
A.~B. Saleh, O.~Bulakci, S.~Redana, B.~Raaf, and J.~Hamalainen, ``Enhancing
  {LTE}-advanced relay deployments via biasing in cell selection and handover
  decision,'' in {\em 2010 IEEE 21st International Symposium on Personal Indoor
  and Mobile Radio Communications (PIMRC)}, pp.~2277--2281, Sep. 2010.

\end{thebibliography}


\end{document}